\newif\ifels
\elsfalse

\ifels
\documentclass[preprint,sort&compress]{elsarticle}
\renewcommand{\cite}{\citep}
\else
\documentclass[a4paper]{article}
\fi

\usepackage{amsmath}
\usepackage{amsfonts}
\usepackage{amssymb}         
\usepackage{amsopn}
\usepackage{theorem}          
\usepackage{latexsym}
\usepackage{graphicx}        
\usepackage{mathrsfs}		
\usepackage{psfrag}
\usepackage{algorithmic}
\usepackage{algorithm}
\usepackage{enumerate}
\usepackage{color}
\usepackage{verbatim}
\usepackage{url}
\usepackage{pgf,tikz}
\usetikzlibrary{arrows}

\definecolor{qqqqff}{rgb}{0.,0.,1.}
\definecolor{ccqqqq}{rgb}{0.8,0.,0.}
\definecolor{zzttqq}{rgb}{0.6,0.2,0.}
\definecolor{ffqqqq}{rgb}{1.,0.,0.}
\definecolor{qqqqff}{rgb}{0.,0.,1.}

\ifels
\else
\fi

\newtheorem{result}{\ }[section]
\theoremstyle{changebreak}                
\newtheorem{thm}[result]{Theorem}

\newtheorem{lem}[result]{Lemma}
\newtheorem{rem}[result]{Remark}
\newtheorem{cor}[result]{Corollary}
\newtheorem{prop}[result]{Proposition}
\newtheorem{eg}[result]{Example}
\newenvironment{proof}
 {{\sl Proof.}\hspace*{1 ex}}%
 {{\nopagebreak\hspace*{\fill}$\Box$\par\vspace{12pt}}}
\newcommand{\transpose}[1]{{#1}^{\top}}

\ifels

\else

\fi

\newcommand{\qedd}{\hfill \ensuremath{\blacksquare}}
\newcommand{\diag}[1]{\mathsf{diag}(#1)}
\newcommand{\rank}[1]{\mathsf{rank}(#1)}

\newcommand{\leo}[1]{{\color{black}#1}}
\begin{document}

\ifels
\else
\thispagestyle{empty}
\begin{center}   
{\LARGE An impossible utopia in distance geometry}
\par \bigskip
{\sc Germano Abud${}^1$, Jorge Alencar${}^2$, Carlile Lavor${}^3$, Leo Liberti${}^{4}$, Antonio Mucherino${}^5$} 
\par \bigskip
\begin{minipage}{15cm}
\begin{flushleft}
{\small
  \begin{itemize}
  \item[${}^1$] {\it FAMAT, Federal University of Uberl\^andia, Minas Gerais, Brazil}  \\ Email:\url{germano.abud@ufu.br}
  \item[${}^2$] {\it IFTM, Federal Institute of Triângulo Mineiro, Brazil} \\ Email:\url{jorgealencar@iftm.edu.br}
  \item[${}^3$] {\it IMECC, University of Campinas, Brazil} \\ Email:\url{clavor@ime.unicamp.br} 
  \item[${}^4$] {\it LIX CNRS, \'Ecole Polytechnique, Institut Polytechnique de Paris, F-91128 Palaiseau, France} \\ Email:\url{liberti@lix.polytechnique.fr}
  \item[${}^4$] {\it IRISA and University of Rennes I, Rennes, France} \\ Email:\url{antonio.mucherino@irisa.fr}
\end{itemize}
}
\end{flushleft}
\end{minipage}
\par \medskip \today
\end{center}
\par \bigskip
\fi

\begin{abstract} 
  The Distance Geometry Problem asks for a realization of a given weighted graph in $\mathbb{R}^K$. Two variants of this problem, both originating from protein conformation, are based on a given vertex order (which abstracts the protein backbone). Both variants involve an element of discrete decision in the realization of the next vertex in the order using $K$ preceding (already realized) vertices. The difference between these variants is that one requires the $K$ preceding vertices to be contiguous. The presence of this constraint allows one to prove, via a combinatorial counting of the number of solutions, that the realization algorithm is fixed-parameter tractable. Its absence, on the other hand, makes it possible to efficiently construct the vertex order directly from the graph. Deriving a combinatorial counting method without using the contiguity requirement would therefore be desirable. In this paper we prove that, unfortunately, such a counting method cannot be devised in general.
  \ifels
  \begin{keyword}
  \else
  \\{\bf Keywords}:
  \fi
  DGP, DMDGP, DDGP, Branch-and-Prune, partial reflection, solution symmetry.
  \ifels
  \end{keyword}
  \fi  
\end{abstract}


\section{Introduction}
\label{s:intro}
We consider the following problem \cite{ddgp,dgbook}, which arises in the determination of protein structure from distance data \cite{muchbook,dgp-sirev}, as well as in the study of rigid graphs constructed by ``Henneberg type 1 moves'' \cite{henneberg1911,tay-whiteley}:
\begin{quote}
  {\sc Discretizable Distance Geometry Problem} (DDGP). Given an integer $K>0$, a simple undirected graph $G=(V,E)$ with an edge weight function $d:E\to\mathbb{R}_+$, and a vertex order $<$ on $V=(1,\ldots,n)$ such that:
  \begin{enumerate}[(i)]
  \item $G[U_0]$ (the subgraph of $G$ induced by $U_0$) is a clique of size $K$, where $U_0=\{1,\ldots,K\}$
  \item $\forall j\in\{K+1,\ldots,n\}\ \exists U_j\subseteq \{1,\ldots,j-1\} \quad |U_j|=K\land \forall i\in U_j \ \{i,j\}\in E$,
  \end{enumerate}
  determine if there is an embedding $x:V\to\mathbb{R}^K$ such that:
  \begin{equation}
    \forall \{i,j\}\in E \quad \|x_i - x_j\|_2^2 = d_{ij}^2. \label{dgp}
  \end{equation}
\end{quote}

The DDGP is a subclass of the more general {\sc Distance Geometry Problem} (DGP) \cite{dgpbook,dgbook}: given $K,G,d$ as above, determine if there is a realization $x$ satisfying Eq.~\eqref{dgp}. An embedding satisfying Eq.~\eqref{dgp} is called a {\it realization}. With a slight abuse of notation we shall also refer to an ``invalid realization'' to denote an embedding which does not satisfy Eq.~\eqref{dgp}, as well as, pleonastically, to a ``valid realization''.
\leo{We also note that the vertex order and the sequence of sets $U_j$ need not be unique.}

Note that a realization $x$ in $\mathbb{R}^K$ of a graph on $n$ vertices can be represented as an $n\times K$ matrix. The $n\times n$ symmetric zero-diagonal matrix having $\|x_i-x_j\|_2^2$ as its $(i,j)$-th entry is a {\it squared Euclidean Distance Matrix} (EDM). It turns out that $G=-\frac{1}{2}JDJ$ (where $J=I_n-\frac{1}{n}\mathbf{1}\transpose{\mathbf{1}}$ is the {\it centering matrix} and $\mathbf{1}$ is the all-one vector) is the {\it Gram matrix} of the realization $x$, i.e.~$x\transpose{x}=G$ \cite{schoenberg,dattorro}. Moreover, $D=\diag{G}\transpose{\mathbf{1}}-2G+\mathbf{1}\transpose{\diag{G}}$, which implies that $\rank{D}\le\rank{G}+2$ \cite{vetterli}; since $\rank{G}=\rank{x}\le K$, we also have $\rank{D}\le K+2$.

Given a realization $x$, we can compute the corresponding EDM by evaluating all Euclidean distances between $x_i,x_j$. Given an EDM $D$ for $x$, we can compute a valid realization by obtaining the Gram matrix $G$ in function of $D$ as explained above, and then factoring $G$ using the spectral decomposition $G=\transpose{P}\Lambda P$, where $P$ is a matrix of eigenvectors and $\Lambda$ a diagonal matrix of corresponding eigenvalues. Then $y=\transpose{P}\sqrt{\Lambda}$ is a valid realization of $D$ (not necessarily equal to $x$).

As mentioned above, the DDGP is a subclass of the DGP including all instances having a vertex order that ensures that the first $K$ vertices form a clique in $G$, and for each remaining vertex $j$ there is a set $U_j$, of $K$ vertices, each of which is adjacent to and precedes $j$ in the given order. This structure allows the application of a certain geometric operation called {\it trilateration} \cite{dgbook} (see Sect.~\ref{s:trilateration} below). Trilateration determines, almost surely, at most two positions $x_j^+,x_j^-\in\mathbb{R}^K$ for vertex $j$ using the distances $d_{ij}$ for each $i\in U_j$. We remark that, when generalized to arbitrary $K$, trilateration is sometimes called $K$-lateration. Moreover, it takes polynomial time in $K$ \cite{alencar2}. Since $K$ is usually fixed in applications, it takes constant time.

Trilateration is a construction also known as ``Henneberg type 1'' \cite{henneberg1911,tay-whiteley}, which entails that all DDGP graphs are rigid. In particular, they have a finite number of incongruent realizations. This follows by definition of rigidity: every isometric continuous motion of a subset of vertices must involve all vertices, and hence be a congruence.

The DDGP is also a super-class of the {\sc Discretizable Molecular Distance Geometry Problem} (DMDGP), which requires each $U_j$ to consist of the $K$ immediate predecessors of $j$. The DMDGP \cite{yemini,dmdgp}, the DDGP \cite{ddgp} and the DGP \cite{saxe79} are all $\mathbf{NP}$-complete. Given a DGP instance, recognizing whether it is DDGP is known as the {\sc Trilateration Ordering Problem} (TOP); recognizing whether it is DMDGP is known as the {\sc Contiguous TOP} (CTOP). It turns out that TOP is $\mathbf{NP}$-complete, but it is in $\mathbf{P}$ for every fixed $K$, whereas CTOP is $\mathbf{NP}$-complete even for fixed $K$ \cite{orders-dam}.  For many protein graphs, however, it is possible to construct a contiguous trilateration order efficiently from the protein backbone \cite{jogomdgp,sidechains,neworder}, which makes the DMDGP a practically interesting class \cite{souza}. 

Repeated trilateration applied to DDGP and DMDGP instances yields an exact algorithm (in the real RAM model \cite{blum}), as follows. The realization of the initial clique $G[U_j]$ of size $K$ can be carried out in constant time (assuming $K$ fixed) by trilateration; then for each subsequent $j$ we construct two alternative positions $x^+_j,x^-_j$ (again in constant time by trilateration), and branch on them. We verify whether neither, one of them, or both satisfy the distances to the predecessors of $j$ not in $U_j$ (if any), and prune those which do not. We obtain a tree search, called {\it Branch-and-Prune} (BP) \cite{protti,lln5}, over the set $S$ of possible positions for vertices $\{K+1,\ldots,n\}$. This tree $T$ has width at most $2^{n-K}$ and depth at most $n$. If $T$ has depth $<n$, then no positions could be found for vertex $n$, which means that the instance is NO. Otherwise, the instance is YES; and any sequence $(x_1,\ldots,x_K,\ldots,x^{s_j}_j,\ldots,x^{s_n}_n)$ of positions found by the BP for all vertices in $V$, where $(s_j\;|\;K<j\le n)$ is a sequence of $+,-$, is a realization of $G$ which certifies a YES. We recall that this certificate is only valid in the real RAM model, which describes a computer able to represent real numbers exactly. In practice, we take $d:E\to\mathbb{Q}_+$, perform operations in floating point, and attempt at minimizing numerical errors using a variety of techniques \cite{iBP-conf,muchbook,bipbip,goncalves,jcim}.

We remark that the tree $T$ is a graph defined over $S\subset\mathbb{R}^K$, and is therefore itself naturally embedded in $\mathbb{R}^K$. Limited to the DMDGP only, two invariant groups of the embedding of $T$ were described in \cite{powerof2-conf,powerof2}. Both groups are reflection groups. The {\it discretization group} is the invariant group of maximum width trees $T$ with $2^{n-K}$ leaf nodes, where each vertex $j$ is adjacent only to the $K$ predecessors in $U_j$ (and possibly some successors); unsurprisingly, it has cardinality power of two. The {\it pruning group}, a subgroup of the discretization group, is the invariant group of the more general case where vertices $j$ may be adjacent to the predecessors in $U_j$ but also to other predecessors. More surprisingly, the pruning group also has power of two many elements. The simple expressions of the cardinalities of these groups derived in \cite{powerof2} were used to argue that the BP algorithm is Fixed-Parameter Tractable (FPT) \cite{bppolybook}. It also allowed the determination of the number of incongruent solutions \cite{liberti-gsi13}, and of a new ``pruning device'' for the BP algorithm \cite{symmBPjbcb} based on symmetry. On the other hand, it was also shown that random DMDGP instances are unlikely to possess large pruning groups \cite{zoo}, and, in particular, that this likelihood rapidly decreases with size.

The techniques used for the structure determination of the discretization and pruning groups are specific to the DMDGP. No easy extension to the DDGP was found so far using those techniques (see \cite{optlet18b} for an attempt). In this paper, we propose a new theoretical analysis of the number of solutions of the DDGP. Specifically, we show that an {\it a priori} computation (i.e.~before running the BP algorithm on the given instance) of the number of incongruent solutions of a DDGP instance is only possible in those instances for which each set of adjacent predecessors $U_j$ induces a clique of size $K$ in the given graph $G$. For those instances, we prove that the number of incongruent realizations is almost surely a power of two, similarly to the DMDGP.

The rest of this paper is organized as follows. In Sect.~\ref{s:prelim}, we analyse the difference between DMDGP and DDGP, we recall the trilateration operation, and give a formal definition of ``combinatorial counting''. In Sect.~\ref{s:cancount} we prove our impossibility result, and present a simple subclass of the DDGP where combinatorial counting is possible. 

\section{Preliminary notions and definitions}
\label{s:prelim}
We recall that most of the properties discussed above only hold almost surely: this occurs because trilateration may fail to work as expected with probability zero, notably when the points realizing vertices in $U_j$ are not in {\it general position} \cite[p.~20]{graverbook}: if $y$ is a realization of $G$ in general position and $W\subseteq V$ then, for each $W\subseteq V$ with $|W|=h+1$, $y[W]$ spans an affine subspace of dimension $h$.

It is always possible to construct infinite families of instances where the edge weight function $d$ is carefully chosen so that there may be more than two possible positions for vertex $j$ using trilateration \cite{powerof2}. But these families all have measure zero in the set of all DDGP (and DMDGP) instances. The same holds for all of the results in this paper. For brevity, we shall refer to ``probability zero instances'' (those over which trilateration fails) and ``probability one instances'' (the rest) --- also see Sect.~\ref{s:counting} below.

The only difference between DMDGP and DDGP is that the sets $U_j$ of adjacent predecessors must also be immediate in the former case, namely $U_j=\{j-K-1,\ldots,j-1\}$. This directly implies that each $G[U_j]$ must be a clique of size $K$ in $G$, which is the property which made it possible to study the symmetries and number of solutions of the DMDGP using the techniques sketched above. DDGP instances may not have this property, however.

For each $j\in V$ we let $\ell(j)=\max_< U_j$, and $\bar{U}_j=U_j\cup\{j\}$. Moreover, let:
\begin{itemize}
\item $N(j)=\{i\in V\;|\;\{i,j\}\in E\}$ be the {\it neighbourhood} of $j$;
\item $\bar{N}(j)=N(j)\cup \{j\}$.
\end{itemize}
We partition the edge set $E$ into the {\it discretization edges} $E_D=\{\{i,j\}\in E \;|\;i\in U_j\}$ and {\it pruning edges} $E_P=E\smallsetminus E_D$.

\subsection{The trilateration operation}
\label{s:trilateration}
Given $K$ points $\{x_1,\ldots,x_K\}\subset\mathbb{R}^K$ and their distances $d_{i}$ to an unknown point $y\in\mathbb{R}^K$, $y$ can be determined by solving the quadratic system of $K$ equations in $K$ unknowns $y=(y_1,\ldots,y_K)$
\begin{equation}
  \forall i\le K \quad \|x_i-y\|_2^2 = d_i^2.
  \label{trilat}
\end{equation}
The trilateration operation is as follows:
\begin{enumerate}
\item Rewrite Eq.~\eqref{trilat} as $\forall i\le K \ \|x_i\|_2^2 + \|y\|_2^2 - 2 x_i y = d_i^2$.
\item Arbitrarily choose one of these $K$ equations, e.g.~the $K$-th one, and form the system of $K-1$ equations in $K$ unknowns given by the difference of the $i$-th equation with the $K$-th one; this removes the term $\|y\|_2^2$ from all equations, leaving the following (after some rearrangements):
  \begin{equation}
    \forall i<K \  2 (x_i - x_K) y = (\|x_i\|_2^2 - \|x_K\|_2^2) - (d_i^2 - d_K^2), \label{trilatlin}
  \end{equation}
  which is a linear underdetermined system in $y$.
\item We assume that Eq.~\eqref{trilatlin} has full rank $K-1$ with probability one, so we can express $K-1$ of the unknowns in function of the remaining one, which we assume wlog to be $y_K$:
  \begin{equation}
    \forall i<K \ y_i = b_i - B_iy_K,
    \label{trilatlinK}
  \end{equation}
  for some $B,b\in\mathbb{R}^K$ \cite[\S 3.3]{dgbook}.
\item We replace $y_1,\ldots,y_{K-1}$ in $\|x_K-y\|_2^2 = d_K^2$ and obtain a quadratic equation in the single unknown $y_K$. We solve this equation and obtain two solutions $y^+_K,y^-_K$ with probability 1, yielding two positions $y^+,y^-$ for $y$ by using Eq.~\eqref{trilatlinK}. 
\item Finally, we check that $y^+,y^-$ satisfy the original equations Eq.~\eqref{trilat}. If they do, the system has two solutions with probability 1. Otherwise, it is infeasible.
\end{enumerate}
           
We denote by $S_j=\tau(y,U_j)$ the trilateration operation in order to determine the position of vertex $j\le n$ in function of the positions $y[U_j]=(y_{i_1},\ldots,y_{i_K})$. We remark that either $S_j=\varnothing$ or $|S_j|=2$ almost surely. 

\subsection{What we mean by ``counting''}
\label{s:counting}
In the real RAM model, the DDGP problem contains an uncountable number of instances, since the edge weight function $d$ maps to the real numbers. This allows us to make statements with some probability (usually zero or one). The trilateration operation, for example, determines zero or two positions for a vertex with probability one (Sect.~\ref{s:trilateration}). If certain special relations between the edge weights hold, it might also determine a single position, or uncountably many \cite{dvop}. These relations between the edge weights induce relations between the points in the valid realizations of the graph, which turn out to be linear equations such as Eq.~\eqref{trilatlin}.

It is intuitive to think that when one or more edge weights continuously change their values within some small enough interval, the positions of the adjacent vertices typically trace continuous trajectories in space.
\begin{eg}
  \label{eg:triangle}
  Consider a triangle graph over $V=\{1,2,3\}$ with $d_{12}=2$ and $d_{13}=d_{23}\in[1,2]$, embedded in $\mathbb{R}^2$. If $x_1=(0,0)$ and $x_2=(2,0)$, then $x_3$ moves continuously on the segment $(1,0) + t(0,1)$ for $t\in[-\sqrt{3},\sqrt{3}]$ as $d_{13},d_{23}$ move continuously in $\alpha=[1,2]$ (see Fig.~\ref{eg21}, left).
  
  \begin{figure}[!ht]
    \begin{center}
      \includegraphics[width=7cm]{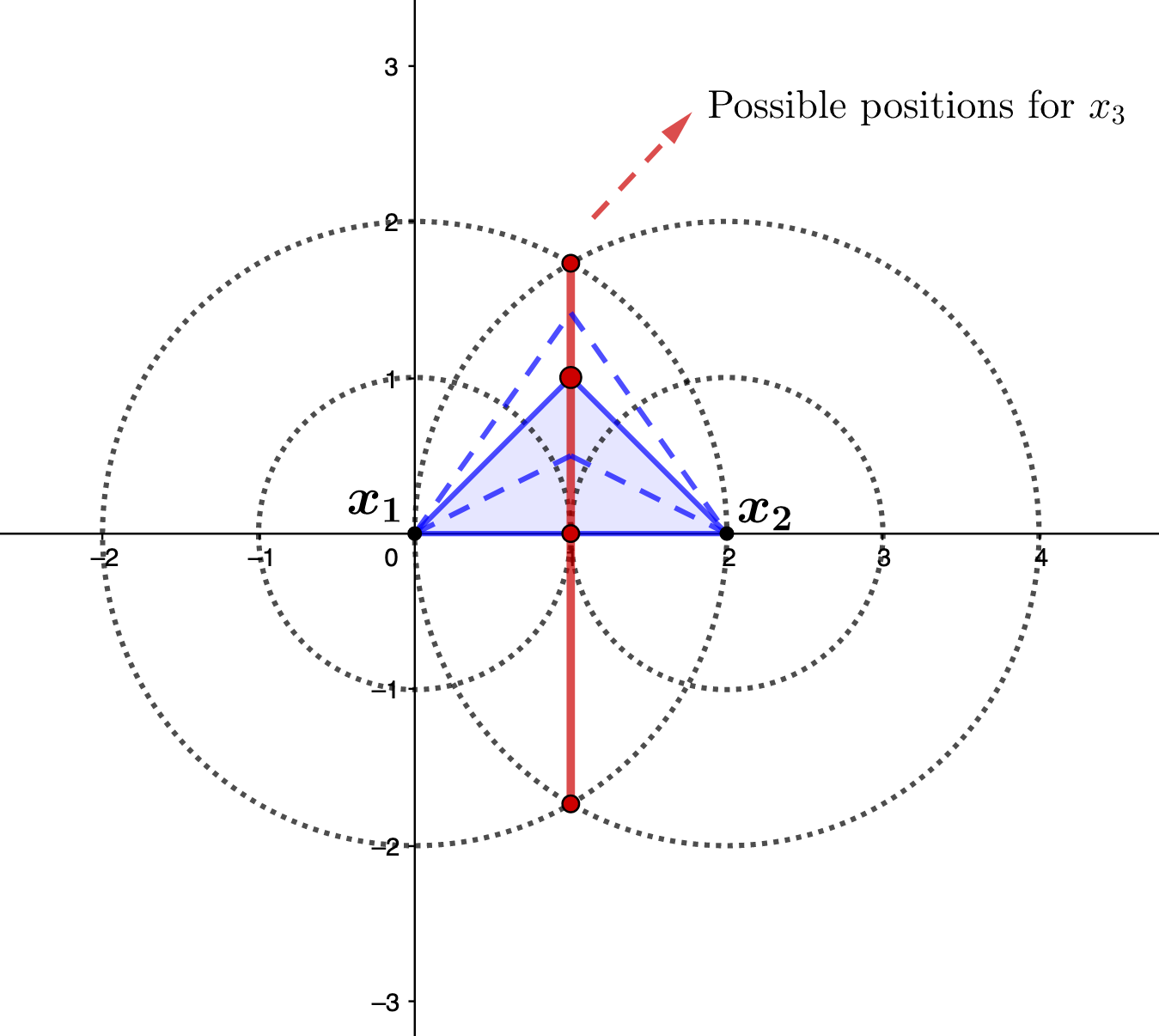}
      \hspace*{0.5cm}
      \includegraphics[width=7cm]{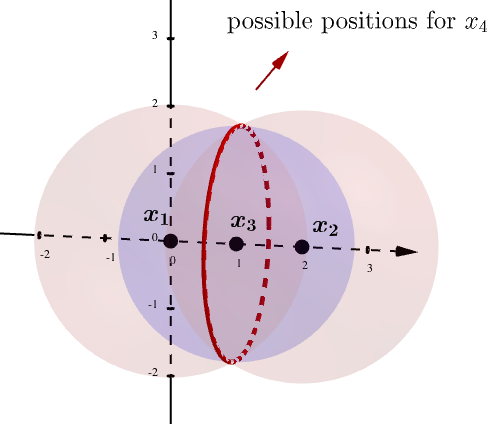}
    \end{center}
    \caption{The two situations depicted in Eg.~\ref{eg:triangle}.}
    \label{eg21}
  \end{figure}
    
  At $t=0$ (corresponding to $d_{13}=d_{23}=1$) the three points $x_1,x_2,x_3$ are aligned, and therefore their affine span has deficient rank equal to $1$: this is a ``probability zero'' realization. All of the other values in the interval define a nontrivial isosceles triangle having full affine span rank $2$. 

  A different choice of $\alpha$ might have yielded an interval where the affine span rank of the associated realization is always full, e.g.~$\alpha=[1.1,2]$. For more complicated graphs it is possible to have situations where both endpoints of the interval yield realizations of deficient ranks.

  Suppose now that we add another vertex (labelled by $4$) to the triangle graph above. We let $4$ be adjacent to $1,2,3$ with edge weights $d_{14}=d_{24}=2$ and $d_{34}=\sqrt{3}$. We consider realizations in $\mathbb{R}^3$. When we apply the trilateration operation to the probability zero realization $x_1=(0,0)$, $x_2=(2,0)$, $x_3=(1,0)$, $x_4$ can move in a circle of radius $\sqrt{3}$ and centered at $(1,0)$. In other words, this trilateration operation finds an uncountable number of positions for $x_4$ (see Fig.~\ref{eg21}, right).  \qedd
\end{eg}
In light of Example \ref{eg:triangle}, we can also define probability zero events over DDGP instances as follows: we construct an uncountable set of DDGP instances where the edge weights are allowed to vary over given intervals, and show that the probability zero event only holds at a finite or countable number of values of the weights in the corresponding intervals.

The goal of this paper is to count realizations of DDGP instances {\it a priori}. The counting methods we consider may not take any feature of the solution into account (for otherwise, the counting problem would be solved by finding all of the finitely many incongruent solutions and counting them). Moreover, we want to avoid events leading to failure of the trilateration operation, such as e.g.~those shown in Example \ref{eg:triangle}. Since these events happen with probability zero, they can be ignored by only considering  {\it combinatorial counting methods}, i.e.~those methods which only consider the graph topology. 

\section{Can we count DDGP realizations combinatorially?}
\label{s:cancount}
In this section we claim that we can only (combinatorially) count realizations for a special subclass of DDGP instances, namely when $U_j$ induces a clique of size $K$ in $G$ for all $K<j\le n$. Our argument is based on the (easier) case of YES instances with no pruning edges. 

For each $j\in\{1,\ldots,n\}$ let $a_j$ be the number of positions, found by the BP algorithm for vertex $j$, which eventually lead to a valid realization of $G$. We assume that the given DDGP instance is YES, and, wlog, that $a_1=\cdots=a_K=1$. Moreover, since the only possible choice for $U_{K+1}$ is $\{1,\ldots,K\}$, which are the immediate predecessors of $K+1$, the DMDGP and DDGP coincide on instances of size $K+1$, which implies that $a_{K+1}=2$ \cite{dmdgp}.

We start with the trivial observation that, by trilateration, there are two positions for vertex $j$ for each position of vertex $\ell(j)$:
  \begin{equation}
    a_j\le 2 a_{\ell(j)}.
    \label{ub}
  \end{equation}
  We now look at conditions which might cause $a_j$ to be strictly less than $2a_{\ell(j)}$, discounting those which hold with probability zero. More precisely, we assume that the given DDGP instance is a probability one instance, and that all realizations of $G$ are in general position. 

Given a realization $y$ of $G=(V,E)$ we let $D^y(V)$ be the EDM of $y$. If $W\subseteq V$ we also let $D^y(W)$ be the EDM of $y[W]$. For brevity we also denote $D^y(W)$ simply by $D(W)$, if no ambiguity should arise in $y$.
\begin{rem}
  If $y$ is a realization of $G$ and $K<j\le n$, then we can write the EDM of $y[\bar{U}_j]$ in the form below:
  \begin{eqnarray}
    D^y(\bar{U}_j) &=& \left(\begin{array}{cc} D^y(U_j) & d^2_{\cdot,j} \\ d^2_{j,\cdot} & 0 \end{array}\right)\label{DM} \nonumber \\
    &=& \left(\begin{array}{cccc|c}
      0 & \|y_{i_1}-y_{i_2}\|_2^2 & \cdots & \|y_{i_1}-y_{i_K}\|_2^2 & d_{i_1,j}^2 \\
      \|y_{i_2}-y_{i_1}\|_2^2 & 0 & \cdots & \|y_{i_2}-y_{i_K}\|_2^2 & d_{i_2,j}^2 \\
      \vdots & \vdots & \ddots & \vdots & \vdots \\
      \|y_{i_K}-y_{i_1}\|_2^2 & \|y_{i_K}-y_{i_2}\|_2^2 & \cdots & 0 & d_{i_K,j}^2 \\ [0.3em] \hline 
      d_{j,i_1}^2 & d_{j,i_2}^2 & \cdots & d_{j,i_K}^2 & 0
    \end{array}\right),\label{DM2}
  \end{eqnarray}
  where $D^y(U_j)$ is expressed in function of $y$, whereas the last row and column is expressed in function of the known edge weights $d$. 
\end{rem}

\begin{lem}
  \label{ifflem}
  Consider a YES DDGP instance and a valid realization $y$ of $G$. If $D^y(U_j)$ is a valid EDM, then $D^y(\bar{U}_j)$ is also a valid EDM.
\end{lem}
\begin{proof}
  Assume $D^y(U_j)$ is a valid EDM but $D^y(\bar{U}_j)$ is not: then $y$ cannot be a valid realization of $G$, against the assumption.
\end{proof}

\begin{prop}
  \label{iffthm}
  Consider a YES DDGP instance, and let $j\in V$ such that $K<j\le n$. If $D^y(U_j)$ is a valid EDM for any possible realization $y$ of $G$, then
  \begin{equation}
    a_j=2 a_{\ell(j)}.
    \label{ubeq}
  \end{equation}
\end{prop}
\begin{proof}
  Let $j\in V$ such that $K<j\le n$. Assume that every realization of $G$ yields a matrix $D(U_j)$ which is a valid EDM. By definition, each of the $a_{\ell(j)}$ possible positions for vertex $\ell(j)$ gives rise to a valid realization $y[U_j]$ of $G[U_j]$. By Lemma \ref{ifflem}, every $D^y(\bar{U}_j)$ is a valid EDM. By trilateration, there are two positions for vertex $j$ for each $y$, which yields Eq.~\eqref{ubeq}.
\end{proof}

\subsection{An impossibility result}  \label{sub:impossibility}
The counterexample below shows what can go wrong if the condition of Prop.~\ref{iffthm} is not met. 
\begin{eg}
  \label{eg:edmfail}
  Consider the graph $G$ on $V=\{1,\ldots,5\}$ with edges
  \[ E=\{\{1,2\},\{1,3\},\{1,5\},\{2,3\},\{2,4\},\{3,4\},\{4,5\}\} \]
  and edge weights $d_{12}=d_{15}=d_{23}=d_{45}=1$, $d_{34}=2$, $d_{13}=\sqrt{2}$, $d_{24}=\sqrt{5}$, realized in $\mathbb{R}^2$. We assume that $x_1=(1,0)$, $x_2=(2,0)$, $x_3=(2,1)$. There are two possible positions for vertex $4$, namely $x_4^+ = (4,1)$, $x_4^-=(0,1)$, as shown in Fig.~\ref{f:edmfail}. However, $\|x_1-x_4^+\|_2=\sqrt{10}$ cannot form a triangle with segments realizing $\{1,5\},\{4,5\}$ both having unit length, since $d_{15}+d_{45}=2<\sqrt{10}$, which negates the triangular inequality on $1,4,5$. On the other hand, the position $x_5=(0,0)$ is compatible with $x_4^-$. 
\begin{figure}[!ht]
\begin{center}
\begin{tikzpicture}[scale=1.0]
\draw[help lines] (-1,-1) grid (3,2);
\draw [->] (-1,0) -- (3,0);
\draw [->] (0,-1) -- (0,2);
\fill (1,0) circle (5pt);
\fill (2,0) circle (5pt);
\fill (2,1) circle (5pt);
\fill (0,1) circle (5pt);
\fill (0,0) circle (5pt);
\draw [ultra thick] (1,0) -- (2,1);
\draw [ultra thick] (1,0) -- (2,0) -- (2,1) -- (0,1) -- (2,0);
\draw [ultra thick] (1,0) -- (0,0) -- (0,1);
\draw (1,-0.4) node(A)  {$x_1$};
\draw (2,-0.4) node(B)   {$x_2$};
\draw (2,1.4) node(C) {$x_3$};
\draw (0,1.4) node(D)  {$x_4^-$};
\draw (0,-0.4) node(4)  {$x_5$};
\end{tikzpicture}
\hspace*{1cm}
\begin{tikzpicture}[scale=1.0]
\draw[help lines] (-1,-1) grid (5,2);
\draw [->] (-1,0) -- (5,0);
\draw [->] (0,-1) -- (0,2);
\fill (1,0) circle (5pt);
\fill (2,0) circle (5pt);
\fill (2,1) circle (5pt);
\fill (4,1) circle (5pt);
\draw [ultra thick] (1,0) -- (2,1);
\draw [ultra thick] (1,0) -- (2,0) -- (2,1) -- (4,1) -- (2,0);
\draw (1,-0.4) node(A)  {$x_1$};
\draw (2,-0.4) node(B)   {$x_2$};
\draw (2,1.4) node(C) {$x_3$};
\draw (4,1.4) node(D)  {$x_4^+$};
\end{tikzpicture}
\end{center}
\caption{Only the realization with $x^-_4$ is feasible (left). The one with $x^+_4$ is not.}
\label{f:edmfail}
\end{figure}
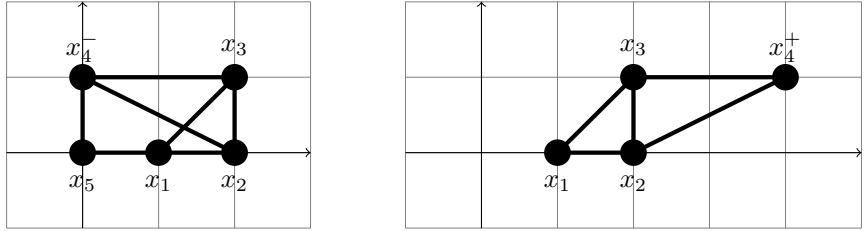
In this case, trilateration would return $S_4$ as the singleton $\{x^-_4\}$, rather than ensuring $|S_4|\in\{0,2\}$ as expected. Note that the above instance is not a ``probability zero instance'', as all $U_j$'s are realized in general position. Generalizations of this counterexample can be obtained for all $K$. \qedd
\end{eg}
The counterexample in Ex.~\ref{eg:edmfail} showcases the necessity of the condition that each matrix $D^y(\bar{U}_j)$ needs to be a valid EDM. Verifying this condition involves checking that all of the matrices $\Gamma_j=-\frac{1}{2}JD^y(\bar{U}_j)J$ are Gram. By the equivalence of Gram and positive semidefinite (psd) matrices, this is equivalent to verifying that all of the $\Gamma_j$'s are psd.

\begin{thm}
  \label{imposs}
  The solutions of the DDGP cannot be counted combinatorially.
\end{thm}
\begin{proof}
  Our definition of combinatorial counting (Sect.~\ref{s:counting}) states that acceptable methods may not consider the edge weights. We now construct an uncountable family of DDGP instances for which trilateration finds 0,1 or 2 positions for a certain vertex, all with positive probability. This shows that the edge weights must necessarily be taken into account by any counting method, and hence that this counting method cannot be combinatorial. We consider the case of Example \ref{eg:edmfail}: our strategy is to define intervals for $d_{24}$ and $d_{34}$ such that: (i) at the lower extrema trilateration on $5$ finds two valid positions for $x_5$, (ii) at the upper extrema trilateration on $5$ only finds one valid position (and hence fails) for $x_5$, and (iii) there are neighbourhoods of these extrema for which the same behaviours hold. This will show that the probability of trilateration failure to find either $0$ or $2$ positions is nonzero, and depends on the edge weights only. Therefore there can be no general combinatorial counting method for dealing with the totality of DDGP instances.

  In the rest of the proof (which simply consists of a long but easy symbolic calculation) we sometimes indicate distance between two vertices $u,v$ by $\overline{uv}$ for brevity. We generalize the instance in Example \ref{eg:edmfail} to the uncountable family of instances given by $d_{24}\in[\sqrt{1+\varepsilon^2},\sqrt{5}]$ and $d_{34}\in[\varepsilon,2]$, for some small enough $\varepsilon>0$. If we take the lower extrema of both intervals $d_{24}=\sqrt{1+\varepsilon^2}$ and $d_{34}=\varepsilon$ we obtain $x^+_4=(2+\varepsilon,1)$ and $x^-_4=(2-\varepsilon,1)$, whence
  \begin{eqnarray*}
    \overline{14^+}=\|x_1-x_4^+\|_2 &=& \sqrt{(-1-\varepsilon)^2+1} = \sqrt{2 + 2\varepsilon + \varepsilon^2} \\
    \overline{14^-}=\|x_1-x_4^-\|_2 &=& \sqrt{(-1+\varepsilon)^2+1} = \sqrt{2 - 2\varepsilon + \varepsilon^2}.
  \end{eqnarray*}
  When $\varepsilon$ is negligible, we have $\overline{14^+}\approx \sqrt{2}<2 = 1+1 = \overline{15}+\overline{4^+5}=d_{15}+d_{45}$ and the same for $\overline{14^-}$, which implies that both positions for vertex $4$ yield a distance $\overline{14}$ that satisfies the triangular inequality. As $\varepsilon$ grows, $\overline{14^-}$ decreases, which means that it satisfies the triangular inequality for all values of $d_{24},d_{34}$ in the respective intervals (as verified in Ex.~\ref{eg:edmfail}). We want to find the value of $\varepsilon$ at which $x_4^+$ satisfies the triangular inequality at equality, namely $\overline{14^+}=\overline{15}+\overline{4^+5}=d_{15}+d_{45}=2$. This happens at $\sqrt{2+2\varepsilon+\varepsilon^2}=2$, namely $\varepsilon^2+2\varepsilon-2=0$, i.e.~when $\varepsilon=\frac{-2\pm\sqrt{4+8}}{2}=-1\pm\sqrt{3}$. Since we assumed $\varepsilon>0$, $\varepsilon=-1+\sqrt{3}$ is the only value for which $\overline{14^+}=d_{15}+d_{45}=2$. Thus, the family of DDGP instances under scrutiny has the property that vertex $5$ has two valid positions (almost surely) for $d_{24}\in[1,\sqrt{5-2\sqrt{3}}]$, $d_{34}\in[0,-1+\sqrt{3}]$, only one position ($x_4^-$) for $d_{24}\in[\sqrt{5-2\sqrt{3}},\sqrt{5}]$, $d_{34}\in[-1+\sqrt{3},2]$, and zero positions in the remaining cases where no position for vertex $4$ exists.

  In other words, assuming uniform probability distributions over the two distance intervals for $d_{24},d_{34}$, we have shown that this DDGP instance family has (almost surely) $2p$ solutions (for some $p\in\mathbb{N}$) with probability $P_2=\frac{\sqrt{5-2\sqrt{3}}+\sqrt{3}-2}{\sqrt{5}+1}\approx 0.3$, $p$ solutions with probability $P_1=\frac{\sqrt{5}-\sqrt{5-2\sqrt{3}}-\sqrt{3}+1}{\sqrt{5}+1}\approx 0.08$, and $0$ solutions in the remaining events where $d_{24}$ is towards the lower extremum while $d_{34}$ is towards the upper one and {\it vice versa}, which have joint probability $P_0=1-\frac{P_2+P_1}{\sqrt{5}+1}=\frac{\sqrt{5}+1-(\sqrt{5}-1)}{\sqrt{5}+1}=\frac{2}{\sqrt{5}+1}\approx 0.62$. Note that $P_0,P_1,P_2>0$, as claimed.
\end{proof}
We note that it is also hard to imagine the existence of non-combinatorial counting method which does not require the realizations of $G$ prior to counting (Sect.~\ref{s:counting}): as mentioned above, we need to check that all of the matrices $\Gamma_j$ are psd, which typically requires the knowledge of the entries of $\Gamma_j$, which in turn requires $y[U_j]$, and hence the realizations of $G$, to be known a priori. 

Thm.~\ref{imposs} does not prevent the existence of counting techniques for subclasses of the DDGP, or based on a condition involving other parameters than $G,d,K$ (such as e.g.~the smallest eigenvalue over all $\Gamma_j$ being nonzero, which would make it easy to prove positive semidefiniteness), or taking into account special structures in the pruning edges. 

\subsection{A sufficient condition}
A combinatorial condition making sure that the $D^y(\bar{U}_j)$ are valid EDMs is that $G[U_j]$ should be a clique.
\begin{cor}
  \label{clique}
  Let $j\in V$ such that $K<j\le n$. If $G[U_j]$ is a clique of size $K$ in $G$, then $a_j=2 a_{\ell(j)}$.
\end{cor}
\begin{proof}
  It suffices to remark that, since all of the $y[U_j]$'s are valid realizations of $G[U_j]$, they must satisfy the given distance constraints. Therefore, $D^y(U_j)$ is simply the EDM for the clique $G[U_j]$, which is constant since the distance values are given for all the edges, and does not depend on $y$. Since we are assuming that the DDGP instance is YES, $D(U_j)$ is a valid EDM. For the same reason, $D(\bar{U}_j)$ is also a valid EDM. 
\end{proof}
We also remark that Cor.~\ref{clique} cannot be improved in general terms, for example by asking that $G[U_j]$ is a clique without one or a few edges, since Ex.~\ref{eg:edmfail} portrays a failure when a single edge is missing from the clique on $G[U_5]$. 

This shows that a combinatorial counting of the number of solutions of DDGP instances prior to actually solving the instance is only possible in the special case where all of the $U_j$'s induce cliques of size $K$ in $G$. We call the class of such DDGP instances the {\it combinatorial DDGP}.
\begin{cor}
  \label{ddgpcliquecount}
  For a combinatorial DDGP instance with discretization edges only, the number of incongruent realizations of $G$ is $2^{n-K}$ almost surely. 
\end{cor}
\begin{proof}
  This follows by $a_1=\cdots=a_K=1$, $a_{K+1}=2$, and Cor.~\ref{clique}.
\end{proof}
We remark that Cor.~\ref{ddgpcliquecount} applies to DMDGP instances. This provides an alternative proof to the result that DMDGP instances with discretization edges only have $2^{n-K}$ incongruent solutions almost surely.



\section{Conclusion}
An important property of DMDGP orders,  mainly in applications related to protein conformation,  is that given a DMDGP solution (calculated by any algorithm applied to the DMDGP),  all the others can be obtained just using the DMDGP symmetries \cite{symmBPjbcb}. Whether there are symmetric properties similar to the DMDGP case at least for the combinatorial DDGP remains an open question.

\section*{Acknowledgements}
CL is grateful to FAPESP and CNPq for support. LL is partly supported by the European Union's Horizon 2020 research and innovation programme under the Marie Sklodowska-Curie grant agreement n. 764759 ETN ``MINOA". AM and LL are grateful to ANR for partly supporting this research under PRCI grant ``MultiBioStruct''.

\bibliographystyle{plain}
\bibliography{ddgpsymm2}

\end{document}